\documentclass[%
 reprint,
 amsmath,amssymb,amsthm,
 aps,
 prx,twocolumn,
]{revtex4}

\usepackage[colorlinks=true,linkcolor=blue,citecolor=blue,urlcolor=blue]{hyperref}
\usepackage{graphicx}
\usepackage{comment}
\usepackage{epsfig}
\usepackage{epstopdf}
\usepackage[usenames,dvipsnames]{xcolor}
\usepackage{braket}
\usepackage{amsthm}
\usepackage{mathtools}
\usepackage{enumitem}
\usepackage{dsfont}
\usepackage{bbold}
\usepackage{eurosym}
\usepackage{mdframed}
\usepackage{cases}
\usepackage{adjustbox}
\usepackage{diagbox}
\usepackage{ragged2e}
\usepackage[symbol]{footmisc}

\usepackage[most]{tcolorbox}
\newtcolorbox{mybox}[2][]{
               = {yshift=-8pt},
  colback      = blue!6!white,
  colframe     = blue!1!black,
  halign       = flush left,
  fonttitle    = \bfseries\sffamily,
  colbacktitle = blue!90!black,
  title        = #2,#1,
  enhanced,
}

\newcommand{\be}{\begin{equation}}
\newcommand{\ee}{\end{equation}}
\newcommand{\ba}{\begin{aligned}}
\newcommand{\ea}{\end{aligned}}

\newcommand{\bc}{\begin{center}}
\newcommand{\ec}{\end{center}}
\newcommand{\beq}{\begin{equation}}
\newcommand{\eeq}{\end{equation}}
\newcommand{\beqq}{\begin{equation*}}
\newcommand{\eeqq}{\end{equation*}}
\newcommand{\beqa}{\begin{align}}
\newcommand{\eeqa}{\end{align}}
\newcommand{\barr}{\begin{array}}
\newcommand{\earr}{\end{array}}
\newcommand{\bi}{\begin{itemize}}
\newcommand{\ei}{\end{itemize}}

\newcommand{\Hi}{\mathcal{H}}

\newcommand{\Tr}{\ensuremath{\,\mathrm{Tr}}}

\newtheorem{lem}{Lemma}
\newtheorem{theo}{Theorem}
\newtheorem*{theo*}{Theorem}

\newtheorem*{prot*}{Protocol}

\newtheorem{defi}{Definition}

\begin{document}

\title{Certification of non-Gaussian states with operational measurements}

\author{Ulysse Chabaud$^{1,2}$}
\email{ulysse.chabaud@gmail.com}
\author{Gana\"el Roeland$^3$}
\author{Mattia Walschaers$^3$}
\author{Fr\'ed\'eric Grosshans$^2$}
\author{Valentina Parigi$^3$}
\author{Damian Markham$^{2,4}$}
\author{Nicolas Treps$^3$}

\affiliation{$^1$Universit\'e de Paris, IRIF, CNRS, France}
\affiliation{$^2$Sorbonne Universit\'e, LIP6, CNRS, 4 place Jussieu, Paris F-75005, France}
\affiliation{$^3$Laboratoire Kastler Brossel, Sorbonne Universit\'e, CNRS, ENS-PSL Research University, Coll\`ege de France, 4 place Jussieu, Paris F-75252, France}
\affiliation{$^4$JFLI, CNRS, National Institute of Informatics, University of Tokyo, Tokyo, Japan}

\date{\today}


\begin{abstract}

We derive a theoretical framework for the experimental certification of non-Gaussian features of quantum states using double homodyne detection. We rank experimental non-Gaussian states according to the recently defined stellar hierarchy and we propose practical Wigner negativity witnesses. We simulate various use-cases  ranging from fidelity estimation to witnessing Wigner negativity. Moreover, we extend results on the robustness of the stellar hierarchy of non-Gaussian states. 
Our results illustrate the usefulness of double homodyne detection as a practical measurement scheme for retrieving information about continuous variable quantum states.

\end{abstract}


\maketitle


\section{Introduction}

\noindent Recent years have witnessed many developments that are paving the road towards quantum technologies \cite{Monz1068,PhysRevLett.123.250503,Quantum-Suppremacy,jurcevic2020demonstration}. The physical backbone of these advances often relies on the discrete nature of certain quantum observables, known as the discrete-variable (DV) approach. In contrast, one can also develop quantum technologies based on systems that have a phase space representation, which leads to quantum observables with a continuum of possible measurement outcomes \cite{Braunstein2005,Pfister_2019,bourassa2020blueprint}. This continuous-variable (CV) approach has gradually made its way to the forefront in the development of quantum information protocols by the experimental realisations of large entangled cluster states in optical setups \cite{Larsen369,Asavanant373}, and quantum error correction in superconducting circuits \cite{QuantumErrorCorrectionGKP}.

CV quantum states are represented in phase space using quasiprobability distributions \cite{cahill1969density}, such as Glauber-Sudarshan, Wigner, or Husimi functions. Quantum states are then split into two main categories according to the shape of their representation in phase space: Gaussian states and non-Gaussian states. In CV quantum optics, Gaussian states and processes are well understood and routinely produced and processed experimentally, the deterministic preparation of highly entangled states over a large number of modes being a striking example of what they enable~\cite{Peng2012,Roslund2013,Chen2014,Yoshikawa2016}. These states and processes also feature an elegant theoretical description with the symplectic formalism~\cite{Lloyd2012}. However, their non-Gaussian counterparts are necessary for a large variety of quantum information processing tasks~\cite{Bartlett2002,mari2012,Rahimi-Keshari2016}.

Non-Gaussian states come in many varieties and can have a wide range of properties. The most benign members of this class of states are non-Gaussian mixtures of Gaussian states, which are generally not considered to be of interest for quantum information procession, because they can still be efficiently simulated via classical computers. Another important class of states are those with non-positive Wigner functions, which are required to reach a quantum computational advantage \cite{mari2012}. One can also find states with a positive Wigner function that cannot be written as a mixture of Gaussian states \cite{filip2011detecting}, these are sometimes referred to as quantum non-Gaussian states \cite{genoni2013detecting}. However, these classes are still crude, containing states with widely different properties in their phase-space representation and that are still hard to classify in term of usefulness in quantum protocols. 

To further explore the structure of CV quantum states in the single-mode-case the stellar hierarchy \cite{chabaud2020stellar} was introduced, where each level corresponds to the number of single-photon additions necessary to engineer the states, together with Gaussian unitary operations.

Such non-Gaussian features in CV quantum states are not only challenging to produce experimentally, they are also difficult to detect.  Homodyne tomography \cite{lvovsky2009continuous, Tiunov:20}, based on maximum-likelihood estimation, is a common tool in quantum optics to reconstruct the full quantum state, thus allowing to extract arbitrary features. However, when implementing this method the number of required measurements grows exponentially with the number of modes. Therefore, it is often more interesting to use a benchmark for the desired properties of the state \cite{eisert2019quantum}, rather than performing a full state tomography. A natural additional requirement for such a benchmark is an accompanying confidence interval, which is, again, hard to extract from maximum-likelihood tomography \cite{PhysRevA.95.022107,Blume_Kohout_2010,PhysRevLett.117.010404}. Machine learning methods have successfully been applied to benchmark the Wigner-negativity of highly multimode states \cite{cimini2020neural}, but these methods are most fruitful when good training data are available. In other words, a machine learning algorithm can only recognize the features of a state when it has seen many similar states before.

In this Manuscript, we provide a complementary approach for certifying the presence of certain features of quantum states, based on double homodyne detection \cite{paris1996quantum,ferraro2005gaussian}. This detection scheme is directly connected to the $Q$ function, which can be exploited for retrieving information about continuous variable quantum states~\cite{chabaud2019building} and for the verification of Boson Sampling output states~\cite{Aaronson2013,chabaud2020efficient}. Here we will use the techniques of \cite{chabaud2019building} to estimate quantum state fidelity, and from these estimates we produce a robust benchmark for the stellar rank of the state. This allows us to operationally identify the number of single-photon additions applied on a Gaussian states that is necessary to reach a given fidelity with a given quantum state. Furthermore, we can use the same techniques to construct witnesses for Wigner negativity. By simulating outcomes for double homodyne detection for realistic states (some of which were reconstructed from experimental data), we explore the experimental feasibility of this protocol. We show that with a realistic amount of loss and an achievable number of measurements, we can extract stellar rank and Wigner-negativity with a high degree of confidence.

\section{Retrieving information via phase space}
\label{sec:phasespace}


\subsection{Phase space measurements} 
\label{sec:phasespacemeasurements}

\begin{figure}
	\begin{center}
		\includegraphics[width=\columnwidth]{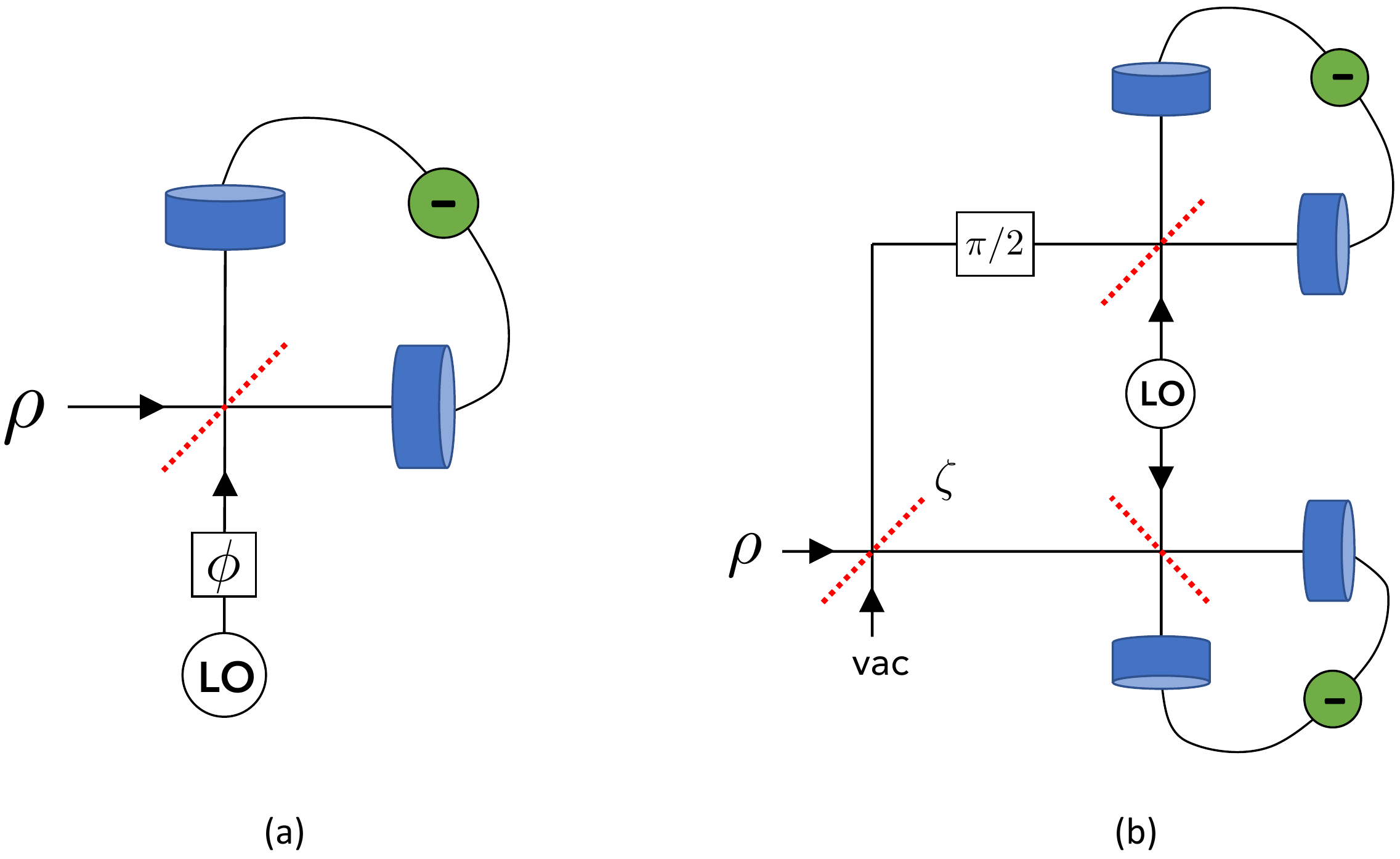}
		\caption{Phase space measurements. \textbf{(a)} Homodyne detection of the rotated quadrature $\hat x_\phi$. \textbf{(b)} Double homodyne detection with unbalancing parameter $\zeta$. The dashed red lines represent beam splitters and LO stands for Local Oscillator, i.e., a strong coherent state.}
		\label{fig:2homodyne}
	\end{center}
\end{figure}

\noindent Homodyne detection is the most widely used experimental procedure to perform quadrature measurements on an optical state. As described in (Fig.~\ref{fig:2homodyne}(a)), it consists in mixing the state to be measured on a balanced beam splitter with a reference field, and compute the difference between intensity measurements on both outputs. Depending on the relative phase $\phi$ between the fields, it yields the sampling of the quadrature observable $\hat x_\phi=\cos\phi\,\hat q+\sin\phi\,\hat p$.


Provided many copies of a state are available, performing optical homodyne detection allows for the reconstruction of the Wigner function of the state. 
This quantum state tomography requires many measurement settings, one for each value of $\phi$, and to solve a complex inverse problem \cite{lvovsky2009continuous}. In practice, a discrete set of angles is chosen, introducing errors in the reconstruction procedure.

Double homodyne detection, sometimes called heterodyne or eight-port homodyne detection~\cite{ferraro2005gaussian}, consists in measuring both quadratures of a single-mode state $\rho$. This is achieved by splitting the state on a beam splitter and performing two homodyne detections, one on each output with a $\pi/2$ phase shift between them (Fig.~\ref{fig:2homodyne}(b)). Double homodyne detection yields a complex outcome, whose real and imaginary parts are respectively both outcomes of the measurements, and projects the measured state onto a coherent state. Hence, it leads to the sampling of the  Husimi $Q$ function of the state, defined as
\be\label{eq:Q}
Q_\rho(\alpha)=\frac1\pi\braket{\alpha|\rho|\alpha},
\ee
for all $\alpha\in\mathbb C$, where $\ket\alpha$ is the coherent state of amplitude $\alpha$.
Note that when the first beam splitter is unbalanced, the double homodyne detection is also unbalanced and projects the measured state onto a finitely squeezed coherent state instead of a coherent state, whose squeezing parameter depends on the unbalancing. In the limit of infinite squeezing, i.e., with a reflectivity of $0$ or $1$ for the input beamsplitter, one retrieves the homodyne detection setup.

Provided many copies of a state are available, optical double homodyne detection allows one to perform continuous variable quantum state tomography with a single measurement setting, since varying the phase of the local oscillator is no longer needed. This is because the Husimi function is directly sampled with double homodyne detection and it contains all the information about the state, whereas for homodyne detection the Wigner function contains all the information about the state but it cannot be sampled directly---it is not a probability density in general---and its marginals are sampled instead.

While being slightly more complicated experimentally than homodyne detection, double homodyne detection has the advantage of providing a reliable tomographic recovery with a single measurement setting~\cite{chabaud2019building}. The POVM element for unbalanced double homodyne detection with outcome $\alpha\in\mathbb C$ is given by
\be
\Pi_{\alpha}=\frac1\pi\ket{\alpha,\zeta}\!\bra{\alpha,\zeta},
\ee
where $\ket{\alpha,\zeta}=\hat S(\zeta)\hat D(\alpha)\ket0$ is a squeezed coherent state, with $\hat D(\alpha)=e^{\alpha\hat a^\dag-\alpha^*\hat a}$ and $\hat S(\zeta)=e^{\frac12(\zeta\hat a^2-\zeta^*\hat a^{\dag2})}$, whose squeezing parameter $\zeta$ depends on the unbalanced double homodyne detection setup. For reflectance $R$ and transmitance $T$ of the beam splitter, we have $|\zeta|=|\log(\frac RT)|$ and the phase of $\zeta$ is given by the phase of the local oscillator~\cite{chabaud2017continuous}. 

In particular, unbalanced double homodyne detection is formally equivalent to a squeezing operation followed by a balanced double homodyne detection. Similarly, since the double homodyne POVM elements are projectors onto coherent states, a displacement before a double homodyne detection can be reverted by translating the classical samples by the amplitude of the displacement. This implies that any single-mode Gaussian unitary operation can be reverted at the level of the measurement by playing with the unbalancing of the detection and by translating the classical samples~\cite{Lloyd2012}.


\subsection{Reliable expectation value estimation with double homodyne detection}
\label{sec:certif}

\noindent A reliable quantum state certification method with double homodyne detection has been introduced in~\cite{chabaud2019building}, and recently enhanced in~\cite{chabaud2020efficient} to allow for the efficient verification of Boson Sampling output states~\cite{Aaronson2013}. Using samples from balanced double homodyne detection of various copies of an experimental (mixed) state $\rho$, this method provides reliable estimates of the expectation value $\Tr(\hat A\rho)$, for any target single-mode operator $\hat A$ with bounded support over the Fock basis.
Choosing $\hat A=\ket l\!\bra k$, one may estimate the density matrix element $\rho_{kl}=\Tr(\ket l\!\bra k\rho)$ in Fock basis and perform a full reconstruction by varying $k$ and $l$. Alternatively, choosing $\hat A=\ket C\!\bra C$, where $\ket C$ is a core state~\cite{menzies2009gaussian,chabaud2020stellar}, i.e., a normalised pure state with bounded support over the Fock basis, one may estimate the fidelity $F(C,\rho)=\Tr(\ket C\!\bra C\rho)$. 

Moreover, since any single-mode Gaussian unitary operation before double homodyne detection can be reverted by unbalancing the detection and translating the classical samples, this also allows for the estimation of expectation values of operators of the form $\hat G\hat A\hat G^\dag$, where $\hat G$ is a Gaussian unitary operation and $\hat A$ is an operator with bounded support over the Fock basis.

This certification method provides analytical confidence intervals for the estimation and makes minimal assumptions on the state preparation: unlike previous existing methods~\cite{paris1996density,paris1996quantum}, it does not assume that the measured state $\rho$ has a bounded support over the Fock basis. This is especially important for the certification of non-Gaussian features, given that such an assumption may induce non-Gaussianity (applying a cutoff in Fock basis to, e.g., a coherent state with nonzero amplitude always makes it non-Gaussian). 
We further assume that identical and identically distributed copies of the measured state $\rho$ are available, but a refined version of the protocol without this assumption is possible \cite{chabaud2019building}.

Let us now describe how to estimate the expectation value of a target operator $\hat A=\sum_{k,l}{A_{kl}\ket k\!\bra l}$, with bounded support over the Fock basis, given a confidence parameter $\delta>0$~\cite{chabaud2019building,chabaud2020efficient}. Considering $N$ copies of the unknown (mixed) quantum state $\rho$, $N$ double homodyne detection are performed leading to the measurement results $\alpha_1,\dots,\alpha_N\in\mathbb C$. One then compute the estimate $F_{\hat A}(\alpha_1,\dots,\alpha_N)$, which is defined as the mean over these measurements of an expectation value estimator whose expression is detailed in Appendix~\ref{app:protocol}. The estimate $F_{\hat A}(\alpha_1,\dots,\alpha_N)$ obtained is $\epsilon$-close to the expectation value $\Tr(\hat A\rho)$ with probability greater than $1-\delta$, for a number of samples $N=\mathcal O(\frac1{\epsilon^{2+t}}\log{\frac1\delta})$, where the prefactor depends on a free parameter $t>0$. This protocol features several free parameters which may be optimised, and we refer to Appendix~\ref{app:protocol} for a detailed analysis of the protocol and its optimisation.

In what follows, we consider applications of this protocol beyond quantum state tomography, for probing non-Gaussian properties of continuous variable quantum states:

\begin{itemize}
\item
Ranking non-Gaussian states (section~\ref{sec:stellar}): due to the robustness of the stellar hierarchy of quantum states~\cite{chabaud2020stellar}, fidelity estimators with non-Gaussian states provide witnesses for the stellar rank. The expectation value estimation protocol yields such fidelity estimators with reliable confidence intervals. We generalise various results of~\cite{chabaud2020stellar} for the robustness of the stellar hierarchy, and illustrate our results with numerical simulations.
\item
Witnessing Wigner negativity (section~\ref{sec:WignerNeg}): since the Wigner function is related to expectation values of displaced parity operators, we obtain witnesses for the negativity of the Wigner function using the double homodyne expectation value estimation protocol.
\end{itemize}


\section{Ranking non-Gaussian states}
\label{sec:stellar}

\noindent A pure state is non-Gaussian if and only if it is orthogonal to a coherent state, which equivalently means that its Husimi $Q$ function (\ref{eq:Q}) has zeros~\cite{lutkenhaus1995nonclassical}. We denote by $\mathcal H_\infty$ the single-mode infinite-dimensional Hilbert space. The stellar rank $r^\star(\psi)$ of a pure single-mode normalised quantum state $\ket\psi\in\Hi_\infty$ has been defined in~\cite{chabaud2020stellar} as the number of zeros of its Husimi function (counted with half multiplicity), where it was shown that it corresponds to the minimal number of photon additions needed to engineer the state from the vacuum, together with Gaussian unitary operations (Fig.~\ref{fig:hierarchy}). We write $\overline{\mathbb N}=\mathbb N\cup\{+\infty\}$, with the convention $n<+\infty\Leftrightarrow n\in\mathbb N$. States of zero stellar rank thus are pure Gaussian states and the stellar rank is either finite or infinite. States having infinite stellar rank include cat states and GKP states (for the latter, the Husimi function has an infinite number of zeros since it is quasiperiodic and GKP states are non-Gaussian).

For a mixed state $\rho$, this definition is generalised to 
\be
r^\star(\rho)=\inf_{\{p_i,\psi_i\}}{\sup_i{\,r^\star(\psi_i)}},
\label{rankmixed}
\ee
where the infimum is taken over the statistical ensembles $\{p_i,\psi_i\}$ such that $\rho=\sum_i{p_i\ket{\psi_i}\bra{\psi_i}}$. In particular, a mixed state has a nonzero stellar rank if it cannot be expressed as a mixture of pure Gaussian states, i.e., when it is quantum non-Gaussian \cite{genoni2013detecting}.

\subsection{Stellar hierarchy of non-Gaussian states}

\begin{figure}
	\begin{center}
		\includegraphics[width=\columnwidth]{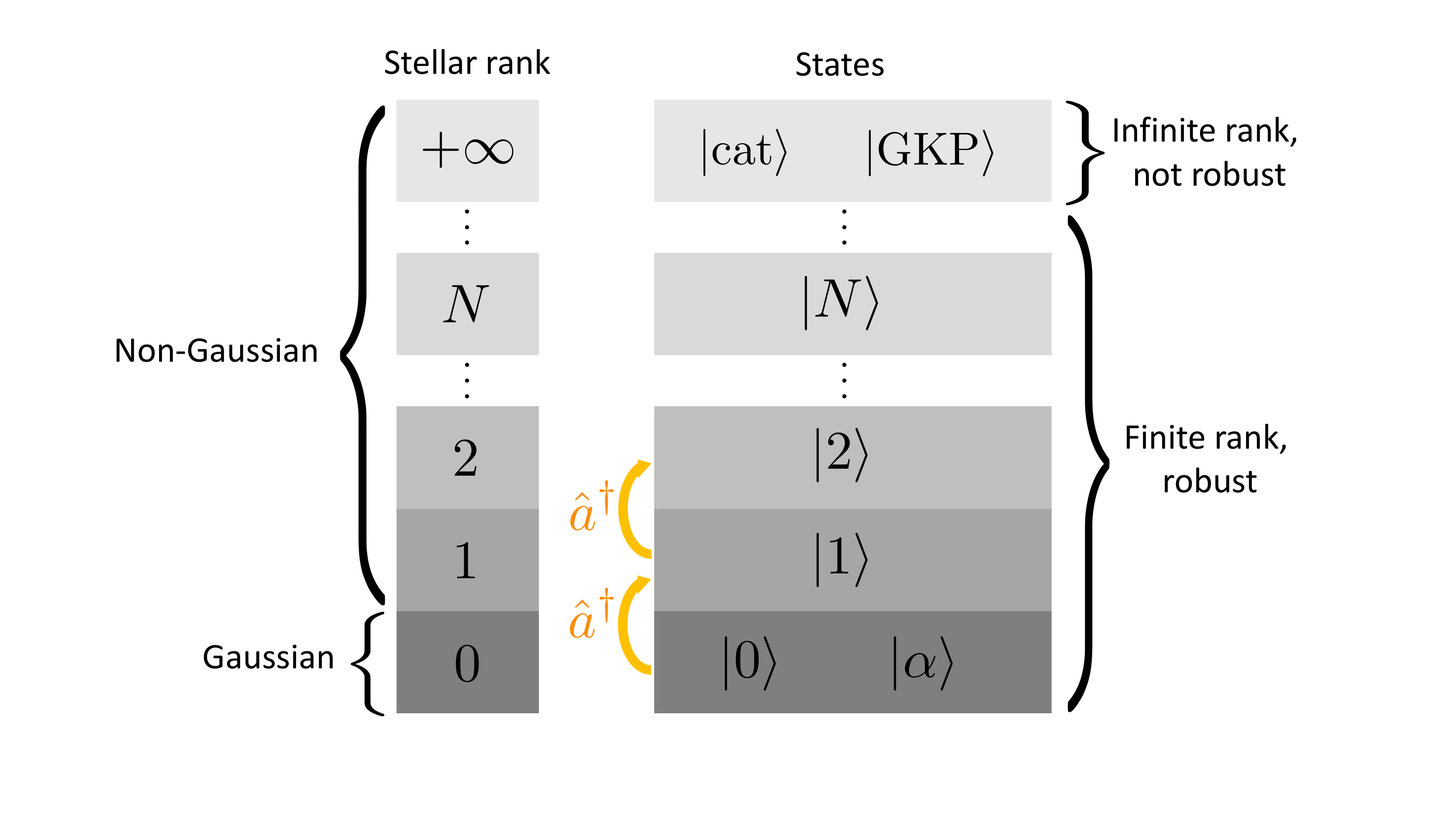}
		\caption{A pictorial depiction of the stellar hierarchy. The right column gives examples of pure states of each stellar rank, which corresponds to the minimal number of photon additions $\hat a^\dag$ necessary to obtain the state from the vacuum. All states of finite stellar rank are robust, i.e., they only have states of equal or higher stellar rank in their close vicinity.}
		\label{fig:hierarchy}
	\end{center}
\end{figure}

\noindent The stellar rank induces a hierarchy among normalised quantum states in $\mathcal H_\infty$, the so-called stellar hierarchy. 

Pure states of finite stellar rank in the hierarchy may be obtained from a core state, i.e., a state with finite support over the Fock basis, with a Gaussian unitary operation~\cite{chabaud2020stellar}. Since single-mode Gaussian unitary operations can be decomposed as a squeezing and a displacement~\cite{Lloyd2012}, pure states of stellar rank $n\in\mathbb N$ are of the form
\be
\hat S(\xi)\hat D(\alpha)\sum_{m=0}^n{c_m\ket m},
\ee
where $\xi,\alpha\in\mathbb C$, $c_n\neq0$ and $\sum_{m=0}^n{|c_m|^2=1}$.
Equivalently, these states may be obtained from the vacuum using $n$ single-photon additions together with Gaussian unitary operations. We show in Appendix~\ref{app:subtraction} that a single-photon subtraction may only increase the stellar rank by at most $1$. This implies that such a state cannot be engineered with less than $n$ single-photon additions and/or subtractions, together with Gaussian unitary operations.

Hereafter, we show how to certify the nonzero stellar ranks of an experimental (mixed) quantum state using fidelity estimation with a target pure state. To that end, we derive robustness properties of the stellar hierarchy and introduce fidelity-based witnesses for the stellar rank.

The stellar robustness $R^\star(\psi)$ of a state $\ket\psi\in\mathcal H_\infty$ corresponds to the amount of deviation in trace distance from $\ket\psi$ which is necessary to reach states of lower stellar rank~\cite{chabaud2020stellar}. In order to probe the stellar hierarchy with phase space measurements, we introduce the following definition, which generalises the notion of stellar robustness:

\begin{defi}[$k$-robustness]
Let $\ket\psi\in\Hi_\infty$. For all $k\in\overline{\mathbb N}^*$, the \textit{$k$-robustness} of the state $\ket\psi$ is defined as
\be
R^\star_k(\psi)\stackrel{\mathrm{def}}{=} \smashoperator{\inf_{r^\star(\phi)<k}} {D(\phi,\psi)},
\ee
where $D(\phi,\psi)=\sqrt{1-|\braket{\phi|\psi}|^2}$ denotes the trace distance and where the infimum is over all states $\ket\phi\in\Hi_\infty$ with a stellar rank lower than $k$. 
\end{defi}

\noindent For all $k\in\mathbb N^*$, the $k$-robustness $R^\star_k$ quantifies how much one has to deviate from a quantum state in trace distance to find another quantum state which has a stellar rank between $0$ and $k-1$. When $k=+\infty$, the $\infty$-robustness quantifies how much one has to deviate from a quantum state in trace distance to find another quantum state of finite stellar rank. Note that for $k=r^\star(\psi)$ we recover the stellar robustness.

The stellar robustness $R^\star$ of a state $\ket\psi$ satisfies~\cite{chabaud2020stellar}:
\be
\smashoperator{\sup_{r^\star(\rho)<r^\star(\psi)}} {F(\rho,\psi)}=1-[R^\star(\psi)]^2,
\ee
where $F(\rho,\psi)=\braket{\psi|\rho|\psi}$ is the fidelity. Certifying that an experimental (mixed) state $\rho$ has a fidelity greater than $1-[R^\star(\psi)]^2$ with a given target pure state $\ket\psi$ thus ensures that the state $\rho$ has stellar rank equal or greater than $r^\star(\psi)$. A similar property holds for the $k$-robustness:

%
\begin{equation}\label{lem:Rfide}
\smashoperator{\sup_{r^\star(\rho)<k}} F(\rho,\psi)=1-[R_k^\star(\psi)]^2.
\end{equation}
%

\noindent We refer to Appendix~\ref{app:Rfide} for a short proof. Certifying that an experimental (mixed) state $\rho$ has a fidelity greater than $1-[R^\star_k(\psi)]^2$ with a given target pure state $\ket\psi$ thus ensures that the state $\rho$ has stellar rank greater or equal to $k$. This implies  that such a state cannot be engineered with less than $k$ single-photon additions and/or subtractions, together with Gaussian unitary operations. In particular, the fidelity of an experimental state with any non-Gaussian target state $\ket\psi$ such that $R_k^\star(\psi)>0$ can serve as a witness for a stellar rank $k$, i.e., as a certificate that the experimental state has a stellar rank no less that $k$. 

Given a target state $\ket\psi$ of stellar rank $r^\star(\psi)$, let us detail when $R_k^\star(\psi)>0$. The robustness of the stellar hierarchy \cite{chabaud2020stellar} implies that each state of a given finite stellar rank is isolated from all the lower stellar ranks. On the other hand, no state of a given finite stellar rank is isolated from any equal or higher stellar rank. That is, for any finite-rank state $\ket\psi\in\mathcal H_\infty$, $R_k^\star(\psi)>0$ when $k\le r^\star(\psi)$, whereas $R_k^\star(\psi)=0$ for $k>r^\star(\psi)$.
Furthermore, states of infinite stellar rank are not isolated from lower stellar ranks, unlike states of finite stellar rank. Any state of infinite can therefore be approximated arbitrarily well by finite-rank states. However, these approximations will always come with a finite error, since states of infinite stellar rank are isolated from states of bounded stellar rank. That is, for states $\ket\psi$ of infinite rank, we find that $R_k^\star(\psi)>0$ for all finite values of $k$, and $R_k^\star(\psi)=0$ for $k=\infty$.

Hence, with Eq.~(\ref{lem:Rfide}) the fidelity with any target pure state of a given stellar rank may serve as a witness for any lower stellar rank. This raises the question of how to compute the value of the maximum achievable fidelity with a given target state using states of finite stellar rank less than $k$ (or equivalently its $k$-robustness), for $k\in\mathbb N^*$.
In~\cite{chabaud2020stellar}, it was shown how to obtain such quantities with a numerical optimisation. However, the number of parameters in the optimisation had to grow with $k$, making the optimisation rapidly impractical. 
Hereafter, we show that the maximum achievable fidelities can be obtained with an optimisation over two complex parameters only. 



Let us define, for all $n\in\mathbb N$,
\be
\Pi_n=\sum_{m=0}^n{\ket m\!\bra m},
\ee
the projector onto the subspace of $\mathcal H_\infty$ spanned by the Fock states $\ket0,\dots,\ket n$.

\begin{theo} \label{th:robust}
Let $k\in\mathbb N^*$ and let $\ket\psi\in\Hi_\infty$. Then, the maximum achievable fidelity with the target state $\ket\psi$ using states of finite stellar rank less than $k$ is given by
\be
\smashoperator{\sup_{r^\star(\rho)<k}}F(\rho,\psi)=\sup_{\hat G\in\mathcal G}{\Tr\left[\Pi_{k-1}\hat G\ket\psi\!\bra\psi\hat G^\dag\right]},
\ee
where the supremum is over Gaussian unitary operations. Moreover, assuming the optimisation yields a Gaussian operation $\hat G_0$, an optimal approximating state is
\be
\hat G_0^\dag\left(\frac{\Pi_{k-1}\hat G_0\ket\psi}{\left\|\Pi_{k-1}\hat G_0\ket\psi\right\|}\right)\!\!.
\ee
\end{theo}

\noindent We refer to Appendix~\ref{app:robust} for a detailed proof. Since single-mode Gaussian unitary operations can be decomposed as a squeezing and a displacement, the optimisation is over two complex parameters only, corresponding to this squeezing and this displacement.

\begin{figure}
	\begin{center}
		\includegraphics[width=\columnwidth]{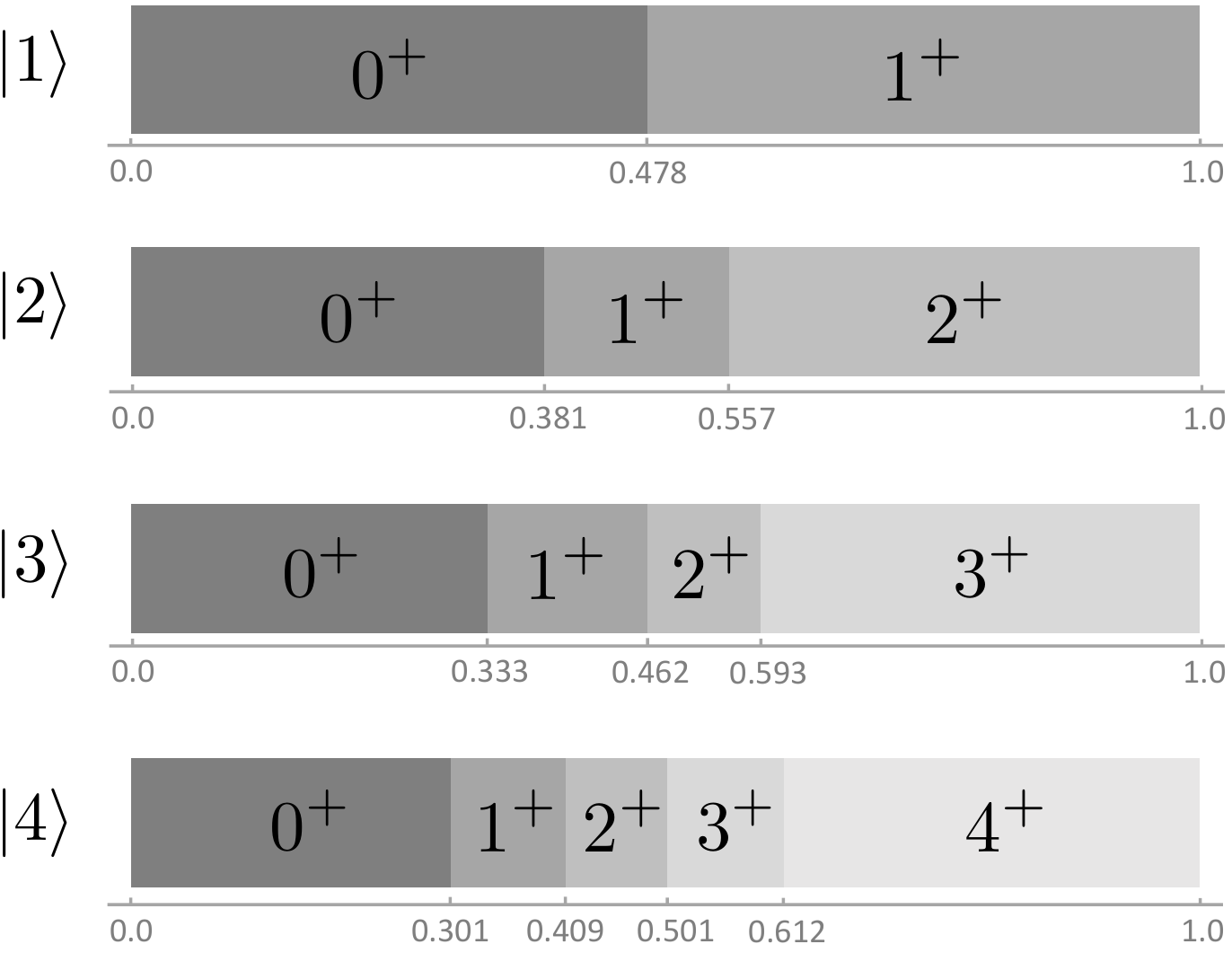}
		\caption{Achievable fidelities with target Fock states $\ket1$, $\ket2$, $\ket3$ and $\ket4$, using states of finite stellar rank. In each fidelity region, the number indicates the minimal stellar rank of states achieving these fidelities, i.e., it corresponds to the minimal number of photon additions and/or subtractions needed to build these states, together with Gaussian operations.}
		\label{fig:RFock}
	\end{center}
\end{figure}

In order to obtain an intuitive graphical representation we introduce the following definition:

\begin{defi}
The profile of achievable fidelities with a non-Gaussian target pure state is defined as the set of achievable fidelities with this state using states of finite stellar rank $k$, for each $k\in\mathbb N^*$.
\end{defi}

\noindent This definition corresponds to the task of engineering an approximation of a target pure state using a finite number of photon additions and/or subtractions, together with Gaussian unitary operations. For each stellar rank, if $\sigma$ denotes a state for which the maximum fidelity is achieved, then any lower fidelity may also be obtained by considering the states $\rho_p=p\ket{\perp}\bra{\perp}+(1-p)\sigma$, for $0\le p<1$, where $\ket{\perp}$ is a coherent state orthogonal to the target state (which exists when the target state is non-Gaussian, i.e., its Husimi function has zeros). In the following, we derive several profiles of achievable fidelities for specific target non-Gaussian states, using Theorem~\ref{th:robust} (see Fig.~\ref{fig:RFock}).

Note that applying a Gaussian unitary operation to the target state yields a state which has the same profile of achievable fidelities, since the $k$-robustness is invariant under these operations. Hence, one may optimise the choice of target state over all Gaussian unitary operations in order to maximise the fidelity of the experimental state with the target state, before making use of the profile of achievable fidelities. In particular, photon-subtracted squeezed states and photon-added squeezed states have the same profile of achievable fidelities as the single-photon Fock state, since they are Gaussian-convertible to that state, i.e., they are related to that state by Gaussian unitary operations~\cite{chabaud2020stellar}.


\subsection{Fock states}

\noindent Let us now consider the example of target core states and Fock states in particular. Using Theorem~\ref{th:robust}, we have obtained numerically the maximum achievable fidelities with a target core state $\cos\phi\ket0+e^{i\chi}\sin\phi\ket1$ of stellar rank $1$ using states of stellar rank $0$, i.e., Gaussian states (see Appendix~\ref{app:core01}). We find that the most robust state of stellar rank $1$ is the single-photon Fock state $\ket1$ (and all the states of the form $\hat G\ket1$, where $\hat G$ is a Gaussian unitary operation).

\begin{figure}
	\begin{center}
		\includegraphics[width=\columnwidth]{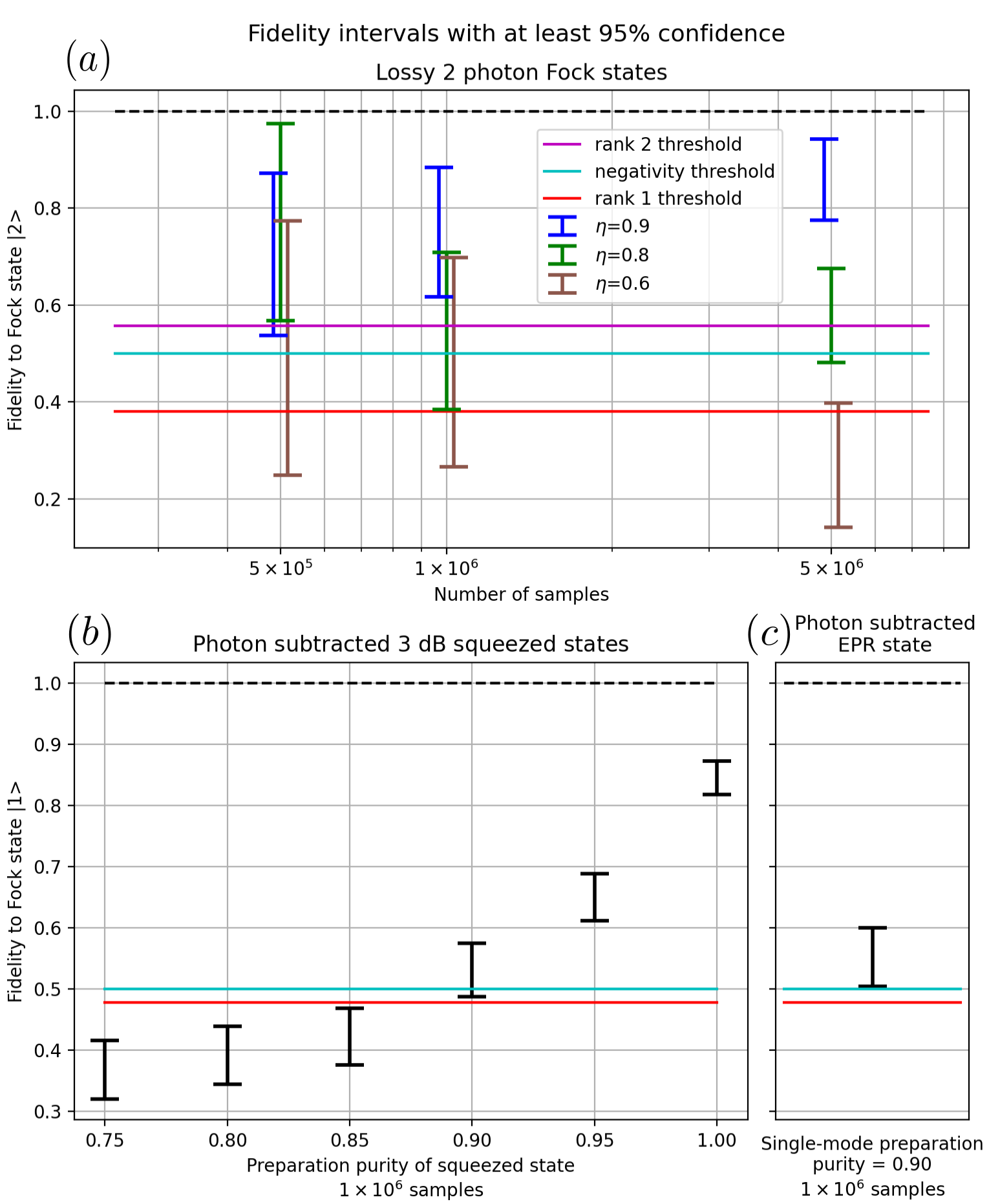}
		\caption{Estimation of fidelity with target Fock states at 95\% confidence, using simulated samples from double homodyne detection of various states. The precision of the estimate varies depending on the number of samples used. We display the stellar rank 1 (resp.\ 2) threshold in red (resp.\ purple) and the Wigner negativity witness threshold (see section \ref{sec:WignerNeg}) in cyan. Note that the rank $1$ threshold varies depending on whether the target state is the Fock state $\ket1$ or $\ket2$, as can be checked with Fig.~\ref{fig:RFock}. \textbf{(a)} Lossy 2 photon Fock states $(1-\eta)^2|0\rangle\langle0|+2\eta(1-\eta)|1\rangle\langle1|+\eta^2|2\rangle\langle2|$, where $\eta$ is an efficiency parameter, and target Fock state $\ket2$. For $\eta=0.9$, $0.8$ and $0.6$, the true fidelities to $|2\rangle$ are $0.81$, $0.64$ and $0.36$, respectively. The Wigner negativity threshold is obtained numerically. \textbf{(b)} Photon-subtracted squeezed vacuum states with $3$dB of squeezing and target Fock state $\ket1$, for several values of the preparation purity, i.e., the purity of the squeezed vacuum state before the photon subtraction. \textbf{(c)} A photon subtracted two-mode EPR state with experimental covariance matrix taken from~\cite{ra2020non} and target Fock state $\ket1$. Simulated photon subtraction is done on the first mode and double homodyne measurement on the second mode of the EPR state.}
		\label{fig:simu_fid}
	\end{center}
\end{figure}

The maximum achievable fidelity with a single-photon Fock state $\ket1$ using states of stellar rank $0$, i.e., Gaussian states, is given by $3\sqrt3/(4e)\approx0.478$~\cite{chabaud2020stellar}. Using Theorem~\ref{th:robust}, we have obtained numerically the profiles of achievable fidelities for Fock states $\ket2$, $\ket3$, $\ket4$ (see Appendix~\ref{app:RFock}), depicted in Fig.~\ref{fig:RFock}. 

These profiles indicate the stellar rank necessary to achieve any fidelity with these Fock states. In particular, the fidelity between an experimental state and a target Fock state can serve as a witness for the stellar rank of the experimental state. For example, we see with Fig.~\ref{fig:RFock} that no state of stellar rank less or equal to $1$ can achieve a fidelity greater than $0.5$ with the Fock state $\ket3$. Hence, if the fidelity of an experimental state with the Fock state $\ket3$ is greater than $0.5$, then it has a stellar rank greater or equal to $2$, implying that it could not have been produced by less than two photon additions or subtractions, together with Gaussian unitary operations.

In order to illustrate the practicality of these fidelity-based witnesses for the stellar rank, we have simulated the double homodyne expectation value estimation from section~\ref{sec:certif} for various states (see Fig.~\ref{fig:simu_fid}). To do so, we computed their $Q$ function and performed rejection sampling to simulate the detection. Then, we estimated their fidelity with target Fock states at 95\% confidence using the simulated samples, thus obtaining witnesses for their stellar rank. For further details on the simulation procedure, see Appendix~\ref{app:protocol}. When the lower bound of the confidence interval is greater than a given threshold, one can certify at 95\% confidence the corresponding property, e.g., a stellar rank greater or equal to~$1$~or~$2$. Note that for Fig.~\ref{fig:simu_fid} \textbf{(c)}, the $Q$ function of the photon subtracted two-mode EPR state is computed from an experimental covariance matrix acquired in~\cite{ra2020non}, where actual measurements of such state display Wigner negativity. Additionally, we have also represented a threshold for the certification of Wigner negativity using a similar witness, as detailed in the next section. 


\section{Wigner negativity}
\label{sec:WignerNeg}

\noindent In the previous sections, we have shown how one can certify the nonzero stellar ranks of experimental (mixed) states, showing in particular that they are non-Gaussian and witnessing their stellar rank. However, such mixed states may still have positive Wigner function. Since processes with positive Wigner functions are classically simulable~\cite{mari2012}, Wigner negativity is a crucial property to look for beyond non-Gaussianity. We show how the double homodyne expectation value estimation protocol from section~\ref{sec:certif} allows for witnessing Wigner negativity in a flexible way without the need for a full tomography. 

The Wigner function of a state $\rho$ evaluated at $\alpha\in\mathbb C$ is related to the expectation value of the parity operator displaced by $\alpha$~\cite{royer1977wigner}:
\be
W_\rho(\alpha)=\frac2\pi\Tr\left[\hat D(\alpha)\hat\Pi_{\mathrm P}\hat D^\dag(\alpha)\rho\right],
\ee
where $\hat\Pi_{\mathrm P}=\sum_{k\ge0}{(-1)^k\ket k\!\bra k}$ is the parity operator. Defining, for all $n\in\mathbb N^*$ and all $\alpha\in\mathbb C$,
\be
\omega_\rho(\alpha,n)\stackrel{\text{def}}{=}\sum_{k=0}^{n-1}{\braket{2k+1|\hat D^\dag(\alpha)\rho\hat D(\alpha)|2k+1}},
\label{witneg}
\ee
we obtain
\be
W_\rho(\alpha)\le\frac2\pi\left[1-2\,\omega_\rho(\alpha,n)\right],
\ee
and this bound becomes tight in the limit of large $n\in\mathbb N^*$. In particular, if $\omega_\rho(\alpha,n)>\frac12$ for any $n\in\mathbb N^*$ and any $\alpha\in\mathbb C$, then the Wigner function of the state $\rho$ is negative at $\alpha$. The expressions $\omega_\rho(\alpha,n)$ thus provide simple witnesses of Wigner negativity, which go beyond those introduced in~\cite{fiuravsek2013witnessing}, since they are tight for large $n$ and provide more flexibility with the choice of $\alpha$. Moreover, we show that they are experimentally accessible in a reliable fashion using double homodyne detection.

Since a displacement before a double homodyne detection can be reverted in post-processing by translating the samples by the amplitude of that displacement (see section~\ref{sec:phasespacemeasurements}), we can estimate the witness $\omega_\rho(\alpha,n)$ for a fixed value of $n$ and $\alpha$, with a precision $\epsilon>0$ and a confidence $\delta>0$, using the expectation value estimation protocol from section~\ref{sec:certif} for the target operator
\be
\hat A_n\stackrel{\text{def}}{=}\sum_{k=0}^{n-1}{\ket{2k+1}\!\bra{2k+1}},
\ee
which has bounded support over the Fock basis, and translating the samples by $\alpha$. If the estimate of the expectation value $\omega_\rho(\alpha,n)=\Tr[\hat D(\alpha)\hat A_n\hat D^\dag(\alpha)\rho]$ is greater than $\frac12+\epsilon$, we may conclude that the Wigner function of the state $\rho$ is negative at $\alpha$, i.e., $W_{\rho}(\alpha)<0$, with probability greater than $1-\frac\delta2$.

\begin{figure}
	\begin{center}
		\includegraphics[width=\columnwidth]{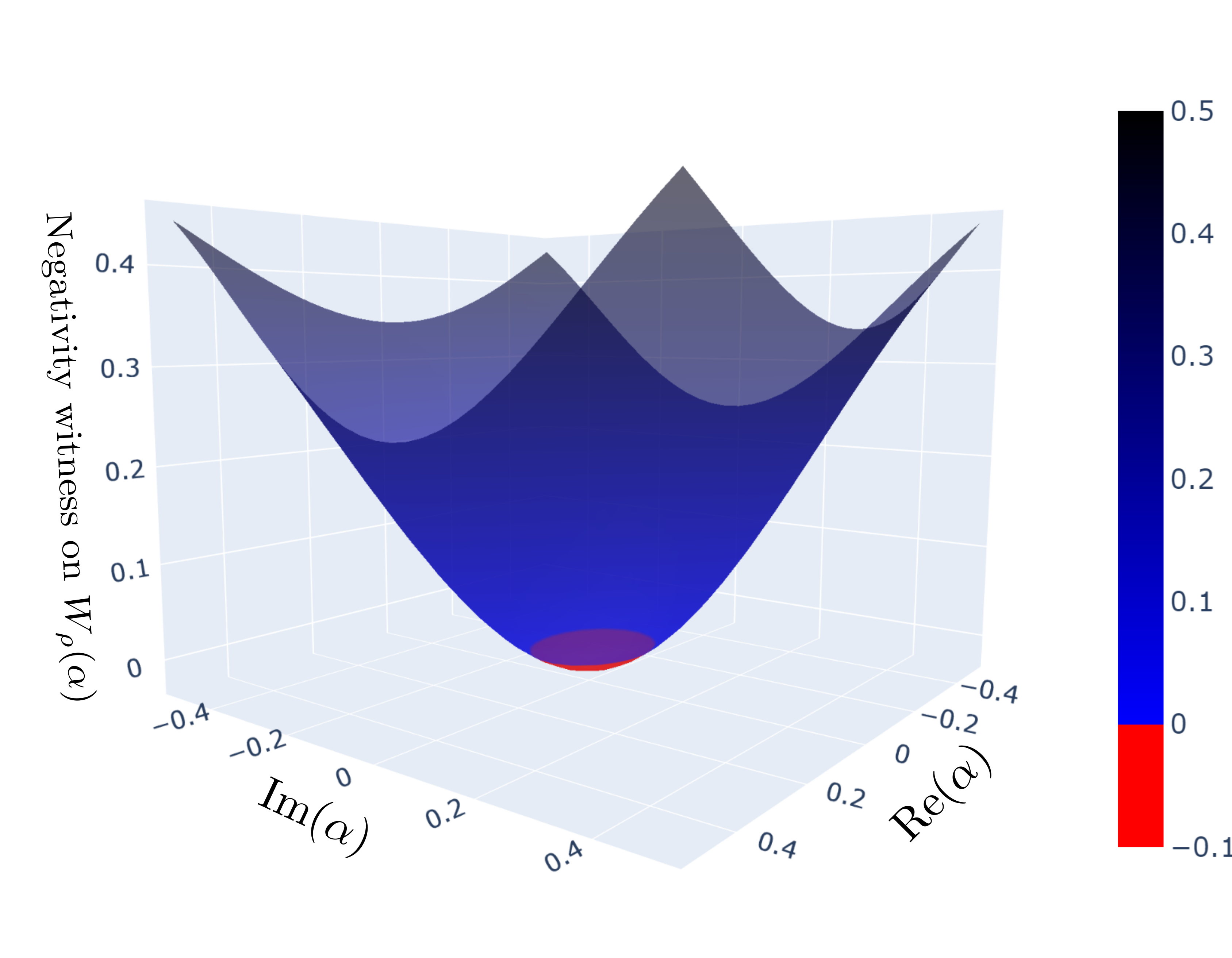}
		\caption{Negativity witnesses for a photon-subtracted squeezed vacuum state with $3$dB of squeezing and $0.95$ preparation purity. Using simulated samples from double homodyne detection, the witnesses in Eq.~(\ref{witneg}) have been estimated for a thousand values of the displacement amplitude $\alpha$, which have been interpolated for clarity, with a fixed precision $\epsilon=0.1$ and a fixed number of samples $N=5.5\times10^5$. The curve has been shifted upwards by the value of $\epsilon$, and the red points witness the negativity of the state with at least $98\%$ confidence.}
		\label{fig:wit_neg}
	\end{center}
\end{figure}

Note that the Wigner negativity witness corresponding to $n=1$ and $\alpha=0$ is simply the fidelity with the single-photon Fock state $\ket1$. If this fidelity is greater than $\frac12$, then the Wigner function is negative at $\alpha=0$ (see Fig.~\ref{fig:simu_fid}, (b) and (c)). 
Using the expectation value estimation from section~\ref{sec:certif} with the target operator $\hat A_1=\ket1\!\bra1$, we have simulated double homodyne sampling of an imperfect photon-subtracted squeezed vacuum state (see Fig.~\ref{fig:wit_neg}), and estimated the corresponding expectation values using the samples translated by $\alpha$, for several values of $\alpha=\alpha_x+i\alpha_y\in\mathbb C$, thus obtaining witnesses of its Wigner negativity.

\medskip
\medskip


\section{Discussion}

\noindent Double homodyne detection is an operational measurement scheme that allows information on CV quantum states to be retrieved and quantum characteristics to be reliably probed. In this work, we obtained new results on the robustness of the stellar hierarchy of non-Gaussian quantum states and in particular on how to compute this robustness for a given target state with a simple optimisation. This allows us to make quantitative and robust statements about non-Gaussian states by deriving a class of witnesses based on the fidelity with specific non-Gaussian target states. Beyond the stellar rank and its operational characterization, we have also obtained new flexible Wigner negativity witnesses, which can be efficiently estimated from experimental data with analytical confidence intervals.

Our work shows that experimental demonstration of the newly introduced witnesses for stellar rank and Wigner negativity is feasible. The avenue towards such an experiment is by itself a direction of future research, since a detailed analysis of a wide range of experimental constraints is still required. Note, for example, that a dedicated analysis on the influence of imperfect double homodyne detection is an important next step.

On the level of theoretical work, a clear line of further research is to obtain robustness profiles of a wider variety of states, such as GKP states or cat states. These profiles will give a direct quantitative understanding of the difficulty of designing such states, as measured by the required number of photon additions and/or subtractions.

The certification technique on which this work is based has recently been extended to the multimode case~\cite{chabaud2020efficient}, allowing the efficient computation of tight fidelity witnesses for a large class of multimode CV states. In combination with the use of entanglement witnesses, we expect that our methods will provide a reliable means of probing the interaction between entanglement and non-gaussianity. These techniques can thus be expected to contribute to the development of an experimental tool for detecting purely non-Gaussian entanglement \cite{PhysRevLett.119.183601}. We leave this fascinating perspective for future work.

\section{Acknowledgements}

D.M.\ and F.G.\ acknowledge funding from the ANR through the ANR-17-CE24-0035 VanQuTe project. V.P.\ acknowledges financial support from the European Research Council under the Consolidator Grant COQCOoN (Grant No.\ 820079).


\bibliography{bibliography}


\widetext
\newpage


\appendix


\section{Double homodyne expectation value estimation protocol}
\label{app:protocol}

\noindent In this section, we first recall the double homodyne expectation value estimation protocol of~\cite{chabaud2019building,chabaud2020efficient} for a target operator $\hat A$ with bounded support over the Fock basis (section~\ref{app:genprot}). This protocol features several free parameters, which we then optimise in the case where the target operator is the density matrix of a Fock state: $\hat A=\ket n\!\bra n$, for $n\in\mathbb N$ (section~\ref{app:optiFock}).


\subsection{General protocol}
\label{app:genprot}

\noindent Following~\cite{chabaud2019building,chabaud2020efficient}, let us introduce for $k,l\ge0$ the polynomials
\be
\ba
\mathcal{L}_{k,l}(z)&=e^{zz^*}\frac{(-1)^{k+l}}{\sqrt{k!}\sqrt{l!}}\frac{\partial^{k+l}}{\partial z^k\partial z^{*l}}e^{-zz^*}\\
&=\sum_{p=0}^{\min{(k,l)}}{\frac{\sqrt{k!}\sqrt{l!}(-1)^p}{p!(k-p)!(l-p)!}z^{l-p}z^{*k-p}},
\ea
\label{2DL}
\ee
for $z\in\mathbb C$, which are, up to a normalisation, the Laguerre 2D polynomials. For all $k,l\in\mathbb N$, we define with these polynomials the functions
\be
\ba
f_{k,l}(z,\eta):=\frac1{\eta^{1+\frac{k+l}2}} e^{\left(1-\frac{1}{\eta}\right)zz^*}\mathcal{L}_{l,k}\left(\frac z{\sqrt{\eta}}\right),
\label{fkl}
\ea
\ee
for all $z\in\mathbb C$ and all $0<\eta<1$. For all $k,l\in\mathbb N$, the function $z\mapsto f_{k,l}(z,\eta)$, being a polynomial multiplied by a converging Gaussian function, is bounded over $\mathbb C$.
Let us define, for all $p\in\mathbb N^*$, all $z\in\mathbb C$ and all $0<\eta<1$,
\be
g_{k,l}^{(p)}(z,\eta):=\sum_{j=0}^{p-1}{(-1)^jf_{k+j,l+j}(z,\eta)\,\eta^j\sqrt{\binom{k+j}k\binom{l+j}l}}.
\label{appg}
\ee
and let
\be
g_{\hat A}^{(p)}(z,\eta):=\sum_{k,l}{A_{kl}\,g_{k,l}^{(p)}(z,\eta)},
\ee
for any operator $\hat A=\sum_{k,l}{A_{kl}\ket k\!\bra l}$ with bounded support over the Fock basis. Let $N\in\mathbb N$, for all $\alpha_1,\dots,\alpha_N\in\mathbb C$, we also define
\be
F_{\hat A}(\alpha_1,\dots,\alpha_N)=\frac1N\sum_{i=1}^N{g_{\hat A}^{(p)}(\alpha_i,\eta)},
\ee
thus omitting the dependencies in the free parameters $\eta$ and $p$. From~\cite{chabaud2020efficient} we have the following result:

\begin{theo*} 
Let $N\in\mathbb N^*$, let $\rho$ be a single-mode state and let $\alpha_1,\dots,\alpha_N\in\mathbb C$ be samples of double homodyne detection of $N$ copies of $\rho$. Let $\epsilon>0$ and $\delta>0$. Then,
\be
\left|\Tr(\hat A\rho)-F_{\hat A}(\alpha_1,\dots,\alpha_N)\right|\le\epsilon
\ee
with probability greater than $1-\delta$, for a number of samples
\be
N=\mathcal O\left(\frac1{\epsilon^{2+t}}\log\left(\frac1\delta\right)\right),
\ee
where $t>0$ is a free parameter depending on the choice of $\eta$ and $p$.
\end{theo*}

\noindent The expressions appearing in the theorem are generic and may be refined depending on the expression of the operator $\hat A$. Moreover, the choice of the free parameters $\eta$ and $p$ can be optimised in order to minimise the number of samples $N$ needed to achieve a given confidence interval $\epsilon,\delta$. We consider such an optimisation in the following section, in the case where $\hat A=\ket n\!\bra n$ is the density operator of a photon number Fock state $\ket n$, for $n\in\mathbb N$.


\subsection{Optimised protocol for target Fock states}
\label{app:optiFock}

\noindent In this section, we optimise the certification protocol from the previous section for the case of photon number Fock states.

Let $0<\epsilon<1$ and $n\in\mathbb N$. In this section, we optimise the single-mode fidelity estimation protocol for a target Fock state $\ket n$. 
For $0<\eta<1$ and $p\in\mathbb N^*$, let
\be
p_n:=\min\left\{q\ge p,\text{ s.t. }\eta\le\left(1-\frac{p-1}q\right)\left(1-\frac n{n+q+1}\right)\right\}.
\label{pn}
\ee
We denote by $\underset{\alpha\leftarrow D}{\mathbb E}[f(\alpha)]$ the expected value of a function $f$ for samples drawn from a distribution $D$. By Lemma~4 of~\cite{chabaud2020efficient} (in particular Eq.~(A17) with $q_0=p_n$ and $k=l=n$),
\be
\left|\underset{\alpha\leftarrow Q_\rho}{\mathbb E}[g_{n,n}^{(p)}(\alpha,\eta)]-\rho_{nn}\right|\le\eta^{p_n}\binom{p_n-1}{p-1}\binom{n+p_n}n.
\label{betterboundEgkl2}
\ee
for any state $\rho$, where $g_{n,n}^{(p)}$ is defined in Eq.~(\ref{appg}), for $k=l=n$. Moreover, by Lemma~4 of~\cite{chabaud2020efficient} we have
\be
\underset{\alpha\leftarrow Q_\rho}{\mathbb E}[g_{n,n}^{(p)}(\alpha,\eta)]=\rho_{nn}+(-1)^{p+1}\sum_{q=p}^{+\infty}{\rho_{n+q,n+q}\eta^q\binom{q-1}{p-1}\binom{n+q}q}.
\ee
In particular, for $p$ odd,
\be
\underset{\alpha\leftarrow Q_\rho}{\mathbb E}[g_{n,n}^{(p)}(\alpha,\eta)]\ge\rho_{nn},
\ee
i.e., $g_{n,n}^{(p)}$ overestimates $\rho_{nn}$, while for $p$ even,
\be
\underset{\alpha\leftarrow Q_\rho}{\mathbb E}[g_{n,n}^{(p)}(\alpha,\eta)]\le\rho_{nn},
\ee
i.e., $g_{n,n}^{(p)}$ underestimates $\rho_{nn}$. Hence, setting
\be
h_n^{(p)}(z,\eta):=g_{n,n}^{(p)}(z,\eta)+\frac12(-1)^p\eta^{p_n}\binom{p_n-1}{p-1}\binom{n+p_n}n,
\label{hn}
\ee
for all $z\in\mathbb C$, we obtain with Eq.~(\ref{betterboundEgkl2}), 
\be
\left|\underset{\alpha\leftarrow Q_\rho}{\mathbb E}[h_n^{(p)}(\alpha,\eta)]-\rho_{nn}\right|\le\frac12\eta^{p_n}\binom{p_n-1}{p-1}\binom{n+p_n}n.
\label{betterboundEh}
\ee
The functions $z\mapsto h_n^{(p)}(z,\eta)$ and $z\mapsto g_{n,n}^{(p)}(z,\eta)$ only differ by a constant and thus have the same range. 
Let $N\in\mathbb N^*$ and let $\alpha_1,\dots,\alpha_N\in\mathbb C$ be i.i.d.\ samples from double homodyne detection of a state $\rho$. For $n\in\mathbb N$, let us write
\be
F_n(\alpha_1,\dots,\alpha_N):=\frac1N\sum_{i=1}^N{h_n^{(p)}(\alpha_i,\eta)},
\ee
where $h_n^{(p)}$ is defined in Eq.~(\ref{hn}).
With Hoeffding inequality~\cite{hoeffding1963probability}, we obtain
\be
\Pr\left[\left|F_n(\alpha_1,\dots,\alpha_N)-\underset{\alpha\leftarrow Q_\rho}{\mathbb E}[h_n^{(p)}(\alpha,\eta)]\right|\ge\lambda\right]\le2\exp\left[-\frac{2N\lambda^2\eta^{2n+2}}{\left(R_n^{(p)}\right)^2}\right],
\label{step1gFock}
\ee
for all $p\in\mathbb N^*$ and all $\eta<1$, where $R_n^{(p)}$ is the range of the function
\be
z\mapsto \eta^{n+1}g_n^{(p)}(z,\eta)=e^{-(1-\eta)\frac{|z|^2}\eta}\sum_{j=0}^{p-1}{(-1)^{n+j}\text{L}_{n+j}\left(\frac{|z|^2}\eta\right)\binom{n+j}n},
\ee
which is bounded for $0<\eta<1$.
Combining Eqs.~(\ref{betterboundEh}) and (\ref{step1gFock}) we obtain
\be
\left|\rho_{nn}-F_n(\alpha_1,\dots,\alpha_N)\right|\le\lambda+\frac12\eta^{p_n}\binom{p_n-1}{p-1}\binom{n+p_n}n,
\ee
with probability greater than $1-P^{iid}_n$, where
\be
P^{iid}_n:=2\exp\left[-\frac{2N\lambda^2\eta^{2n+2}}{\left(R_n^{(p)}\right)^2}\right].
\label{Piidnapp}
\ee
Setting
\be
\lambda=\max\left\{\epsilon-\frac12\eta^{p_n}\binom{p_n-1}{p-1}\binom{n+p_n}n,0\right\},
\ee
we finally obtain
\be
\left|\rho_{nn}-F_n(\alpha_1,\dots,\alpha_N)\right|\le\epsilon
\ee
with probability $1-P^{iid}_n$, where
\be
P^{iid}_n=2\exp\left[-\frac{2N\lambda^2\eta^{2n+2}}{\left(R_n^{(p)}\right)^2}\right],
\label{Piidnapp2}
\ee
where
\be
\ba
\,&\lambda=\max\left\{\epsilon-\frac12\eta^{p_n}\binom{p_n-1}{p-1}\binom{n+p_n}n,0\right\}\\
\,&p_n=\min\left\{q\ge p,\text{ s.t. }\eta\le\left(1-\frac{p-1}q\right)\left(1-\frac n{n+q+1}\right)\right\}\\
\,&R_n^{(p)}=\max_z\left[\eta^{n+1}g_n^{(p)}(z,\eta)\right]-\min_z\left[\eta^{n+1}g_n^{(p)}(z,\eta)\right]\\
\,&\eta^{n+1}g_n^{(p)}(z,\eta)=(-1)^ne^{-(1-\eta)\frac{|z|^2}\eta}\sum_{j=0}^{p-1}{\text{L}_{n+j}\left(\frac{|z|^2}\eta\right)\binom{n+j}n}\\
\,&\text{L}_n(x)=\sum_{i=0}^n{\frac{(-1)^i}{i!}\binom nix^i},\\
\ea
\label{constants}
\ee
for all $z\in\mathbb C$, all $x\in\mathbb R$, all $n\in\mathbb N$, all $p\in\mathbb N^*$, all $0<\eta<1$ and all $0<\epsilon<1$. 

At this point, all expressions are analytical, optimised for the case of a target Fock state. The protocol still has two free parameters $p\in\mathbb N^*$ and $0<\eta<1$, once $\epsilon$, $n$ and $N$ are fixed. The procedure for optimising the efficiency of the protocol then consists in minimising the failure probability over the choice of these parameters $p,\eta$: for a target Fock state $\ket n$ with $n\in\mathbb N$ and a precision parameter $0<\epsilon<1$, we compute for increasing values of $p\in\mathbb N^*$, starting from $p=1$, the minimum failure probability $P^{iid}_n$ in Eq.~(\ref{Piidnapp2}) by optimising over the value of $\eta$; then, we pick the value of $p$ which minimises $P^{iid}_n$. 

For example, for 95\% confidence, i.e., $P^{iid}_n\le0.05$, we obtain in Table~\ref{table:eff1},~\ref{table:eff2} and~\ref{table:eff3} the number of samples $N$ necessary for the fidelity estimate precisions $\epsilon=0.1$, $0.2$ and $0.3$, for the target Fock states $\ket{n=0}$, $\ket{n=1}$ and $\ket{n=2}$, together with the values of the free parameters $p$ and $\eta$ and the corresponding value of $p_n$ (which is needed in the definition of the estimate $h_n^{(p)}$).

\begin{table}[h!]
 \begin{tabular}{||c|c|c|c|c||}
 \hline
$\quad n\quad$ & $N$ & $\quad p\quad$ & $\quad\eta\quad$ & $\quad p_n\quad$\\ [0.5ex] 
 \hline\hline
 $0$ & $\;2.7\times10^4\;$ & $3$ & $0.34$ & $4$\\ 
 \hline
 $1$ & $\;5.5\times10^6\;$ & $3$ & $0.26$ & $3$\\
 \hline
 $2$ & $\;1.3\times10^9\;$ & $3$ & $0.21$ & $3$\\
 \hline
\end{tabular}
\caption{Optimised parameters for target Fock states $0$, $1$ and $2$, for precision $\epsilon=0.1$ with 95\% confidence.}
\label{table:eff1}
\end{table}

\begin{table}[h!]
 \begin{tabular}{||c|c|c|c|c||}
 \hline
$\quad n\quad$ & $N$ & $\quad p\quad$ & $\quad\eta\quad$ & $\quad p_n\quad$\\ [0.5ex] 
 \hline\hline
 $0$ & $\;3.6\times10^3\;$ & $2$ & $0.35$ & $2$\\ 
 \hline
 $1$ & $\;5.8\times10^5\;$ & $2$ & $0.26$ & $2$\\
 \hline
 $2$ & $\;1.0\times10^8\;$ & $3$ & $0.25$ & $4$\\
 \hline
\end{tabular}
\caption{Optimised parameters for target Fock states $0$, $1$ and $2$, for a precision $\epsilon=0.2$ with 95\% confidence.}
\label{table:eff2}
\end{table}

\begin{table}[h!]
 \begin{tabular}{||c|c|c|c|c||}
 \hline
$\quad n\quad$ & $N$ & $\quad p\quad$ & $\quad\eta\quad$ & $\quad p_n\quad$\\ [0.5ex] 
 \hline\hline
 $0$ & $\;9.1\times10^2\;$ & $1$ & $0.30$ & $1$\\ 
 \hline
 $1$ & $\;1.2\times10^5\;$ & $2$ & $0.31$ & $2$\\
 \hline
 $2$ & $\;1.6\times10^7\;$ & $2$ & $0.24$ & $2$\\
 \hline
\end{tabular}
\caption{Optimised parameters for target Fock states $0$, $1$ and $2$, for a precision $\epsilon=0.3$ with 95\% confidence.}
\label{table:eff3}
\end{table}

In order to optimise further the efficiency of the protocol, we may replace the use of the Hoeffding inequality by a direct application of the central limit theorem (CLT). Instead of looking at the range of the distribution $h_n$, one can derive the variance of the sampled distribution. Then, one deduces a gaussian confidence interval of the mean estimation, since the central limit theorem ensures that the distribution of the mean converges towards a gaussian distribution. The theorem is valid for distribution whose variance is finite, which is the case for physical states.

Similarly, we obtain Eq.~(\ref{Piidnapp2}) with probability greater than $1-P^{iid}_{n,CLT}$, where
\be
P^{iid}_{n,\text{CLT}} = 1 - \mathrm{erf}\left(\lambda \sqrt{\frac{N}{2 \sigma^2}}\right),
\ee
$\sigma^2$ is the variance of the sampled estimations, and erf is the error function.

Applying the central limit theorem comes at the cost that the obtained bounds are no longer analytical, since we replace the theoretical variance by its value estimated from the samples.

To emphasize the usefulness of the optimisation carried out in this section, we have represented in Fig.~\ref{fig:opti} the difference in efficiency between the non-optimised protocol and its optimised versions, with Hoeffding inequality or CLT, for a target Fock state $\ket2$. 

\begin{figure}[h!]
	\begin{center}
		\includegraphics[width=0.7\columnwidth]{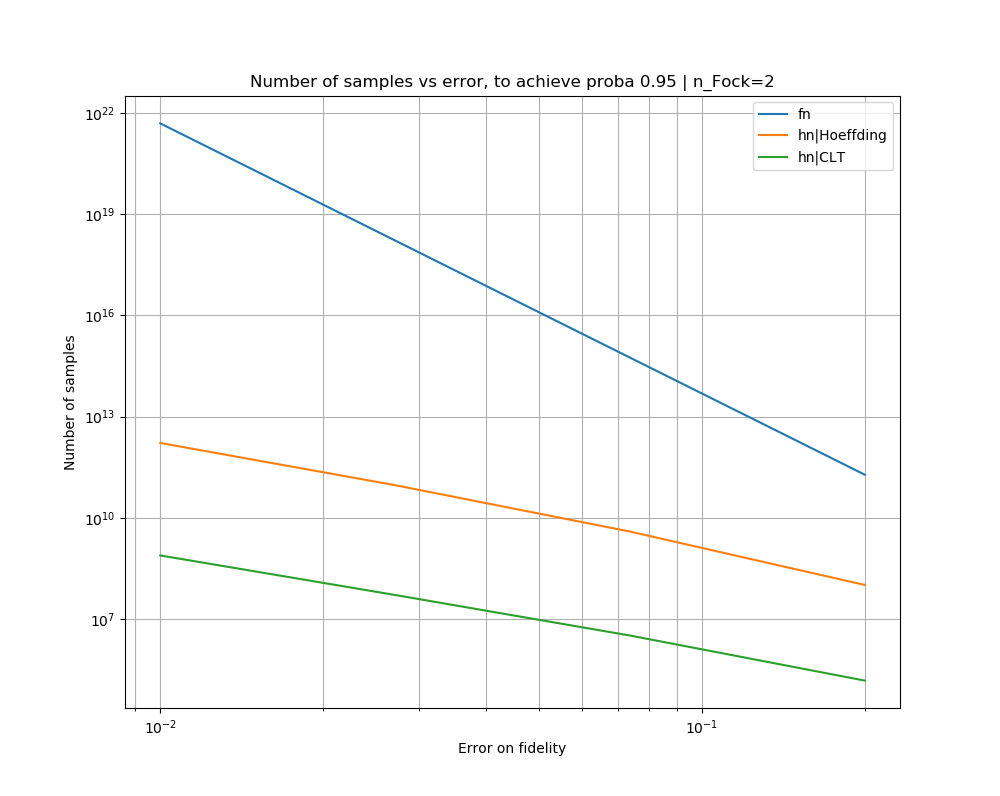}
		\caption{Efficiency of the various estimation methods for a target Fock state $\ket2$. The curves represent the number of samples necessary to obtain a given precision on the fidelity estimate with 95\% confidence. The top curve is obtained using the protocol from~\cite{chabaud2019building}, the middle one using the optimised protocol with analytical bounds, and the bottom one using the improved bounds from the CLT (for which a Fock state $\ket2$ is sampled).}
		\label{fig:opti}
	\end{center}
\end{figure}
%


\section{Photon subtraction and stellar rank}
\label{app:subtraction}

\noindent The stellar function of a single-mode pure state $\ket\psi$ is defined as~\cite{chabaud2020stellar}:
\be
F_\psi^\star(z)=e^{\frac12|z|^2}\braket{\psi|z},
\ee
where $\ket z=e^{-\frac12|z|^2}\sum{n\ge0}{\frac{z^n}{\sqrt{n!}}\ket n}$ is the coherent state of amplitude $z$, for all $z\in\mathbb C$. This function gives a representation of a quantum state with a holomorphic function which is related to the Husimi $Q$ function by
\be
Q_\psi(z)=\frac{e^{-|z|^2}}\pi\left|F_\psi^\star(z)\right|^2.
\ee
The stellar rank of a state $\ket\psi$ is the number of zeros of its stellar function (counted with multiplicity) and corresponds to the minimal number of elementary non-Gaussian operations---namely, photon-addition--- necessary to engineer the state from the vacuum, together with unitary Gaussian operations.

While operators have their own treatment in the stellar formalism, it is sufficient for our purpose to consider the following correspondences: the creation and annihilation operators have the stellar representations~\cite{chabaud2020stellar}
\be
\hat a^\dag\rightarrow z,\quad \hat a\rightarrow\partial_z,
\label{castellar}
\ee
i.e., the operator corresponding to $\hat a^\dag$ in the stellar representation is the multiplication by $z$ and the operator in the stellar representation corresponding to $\hat a$ is the derivative with respect to $z$. This means that the stellar function of a photon-added state $\hat a^\dag\ket\psi$ is given by $z\mapsto zF_\psi^\star(z)$, while the stellar function of a photon-subtracted state $\hat a\ket\psi$ is given by $z\mapsto\partial_zF_\psi^\star(z)$. In particular, photon-added states are always non-Gaussian~\cite{walschaers2017statistical}, since $0$ is a root of their stellar function, while photon-subtracted states can be Gaussian (e.g., the Fock state $\ket1$, for which $\hat a\ket1=\ket0$, or the coherent states $\ket\alpha$, for $\alpha\in\mathbb C$, for which $\hat a\ket\alpha=\alpha\ket\alpha$).

Moreover, by Theorem~1 of~\cite{chabaud2020stellar}, the stellar function of a state of finite stellar rank is a polynomial multiplied by a Gaussian function $\exp(Sz^2+Dz)$, where $S$ is a squeezing parameter and $D$ a displacement parameter. Hence, derivating this function gives another polynomial multiplied by the same Gaussian term. The new polynomial has a degree increased by $1$ if $S\neq0$, the same degree if $S=0$ and $D\neq0$, and a degree decreased by one if $S=D=0$. Hence, a single-photon subtraction can increase the stellar rank by at most $1$ since the number of zeros of the stellar functions is equal to the degree of this polynomial.

For completeness, note that the photon subtraction may conserve an infinite stellar rank (e.g., for a cat state) but may also make it finite (e.g., for a displaced cat state $\propto\ket0+\ket\alpha$).


\section{Proof of Lemma~\ref{lem:Rfide}}
\label{app:Rfide}

\noindent For any pure state $\ket\psi\in\Hi_\infty$ and any set of pure states $\mathcal X$, we have
\be
\ba
\sup_{\substack{\rho=\sum{p_i{\ket\phi_i}\bra{\phi_i}}\\ \sum{p_i}=1,\phi_i\in\mathcal X}}{F(\rho,\psi)}&=\sup_{\substack{\rho=\sum{p_i{\ket\phi_i}\bra{\phi_i}}\\ \sum{p_i}=1,\phi_i\in\mathcal X}}{\braket{\psi|\rho|\psi}}\\
&=\sup_{\sum{p_i}=1}\sup_{\phi_i\in\mathcal X}{\sum{p_i|\braket{\phi_i|\psi}|^2}}\\
&=\sup_{\phi\in\mathcal X}{|\braket{\phi|\psi}|^2}\\
&=\sup_{\phi\in\mathcal X}{F(\phi,\psi)}.
\ea
\label{supF}
\ee
Hence, for $\mathcal X$ the set of pure states of stellar rank less than $k$,
\be
\ba
R_k^\star(\psi)&=\inf_{r^\star(\phi)<k}{D(\phi,\psi)}\\
&=\inf_{r^\star(\phi)<k}{\sqrt{1-|\braket{\phi|\psi}|^2}}\\
&=\sqrt{1-\sup_{r^\star(\phi)<k}{F(\phi,\psi)}}\\
&=\sqrt{1-\sup_{r^\star(\rho)<k}{F(\rho,\psi)}},
\ea
\label{robustnessinter}
\ee
where $D$ denotes the trace distance, where we used the definition of the stellar rank for mixed states~(\ref{rankmixed}). We finally obtain
\be
\sup_{r^\star(\rho)<k}{F(\rho,\psi)}=1-[R_k^\star(\psi)]^2.
\ee

\qed


\section{Proof of Theorem~\ref{th:robust}}
\label{app:robust}

\noindent By Eq.~(\ref{supF}) we have
\be
\sup_{r^\star(\rho)<k}F(\rho,\psi)=\sup_{r^\star(\phi)<k}{|\braket{\phi|\psi}|^2}.
\label{Fkinter}
\ee
By Theorem~4 of~\cite{chabaud2020stellar}, for any pure state $\ket\phi$ such that $r^\star(\phi)<k$, there exist a normalised core state with a finite support over the Fock basis $\ket{C_\phi}$ of stellar rank lower than $k$ and a Gaussian operation $\hat G_\phi$ such that
\be
\ket\phi=\hat G_\phi\ket{C_\phi}.
\label{phiGC}
\ee
We obtain
\be
\ba
|\braket{\phi|\psi}|^2&=|\braket{C_\phi|\hat G_\phi^\dag|\psi}|^2\\
&=|\braket{C_\phi|\Pi_{k-1}\hat G_\phi^\dag|\psi}|^2\\
&\le|\braket{C_\phi|C_\phi}|^2|\braket{\psi|\hat G_\phi\Pi_{k-1}\hat G_\phi^\dag|\psi}|^2\\
&=\Tr\left[\Pi_{k-1}\hat G_\phi^\dag\ket\psi\bra\psi\hat G_\phi\right],
\ea
\label{boundoverlapPik}
\ee
where we used $\ket{C_\phi}=\Pi_{k-1}\ket{C_\phi}$ in the second line, since $\ket{C_\phi}$ is a core state of stellar rank lower than $k$ (hence its support is contained in the support of $\Pi_{k-1}$), Cauchy-Schwarz inequality in the third line and $|\braket{C_\phi|C_\phi}|^2=1$ in the last line. This upperbound is attained if
\be
\ket{C_\phi}=\frac{\Pi_{k-1}\hat G_\phi^\dag\ket\psi}{\sqrt{\Tr\left[\Pi_{k-1}\hat G_\phi^\dag\ket\psi\bra\psi\hat G_\phi\right]}},
\label{coreapprox}
\ee
which is indeed a normalised core state of stellar rank lower than $k$.

With Eqs.~(\ref{Fkinter}) and~(\ref{boundoverlapPik}), the maximum achievable fidelity with the state $\ket\psi$ is then given by
\be
\ba
\sup_{r^\star(\rho)<k}F(\rho,\psi)&=\sup_{\hat G_\phi\in\mathcal G}{\Tr\left[\Pi_{k-1}\hat G_\phi^\dag\ket\psi\bra\psi\hat G_\phi\right]}\\
&=\sup_{\hat G\in\mathcal G}{\Tr\left[\Pi_{k-1}\hat G\ket\psi\bra\psi\hat G^\dag\right]},
\ea
\ee
where the supremum is over Gaussian unitary operations and where we used the fact that the set of Gaussian unitary operations is invariant under adjoint in the second line. With Eq.~(\ref{phiGC}), assuming the optimisation yields an optimal Gaussian unitary $\hat G_0$, an optimal approximating state is $\ket\phi=\hat G_\phi\ket{C_\phi}$, where $\hat G_\phi=\hat G_0^\dag$ and where $\ket{C_\phi}$ is given by Eq.~(\ref{coreapprox}). Namely,
\be
\ket\phi=\hat G_0^\dag\left(\frac{\Pi_{k-1}\hat G_0\ket\psi}{\left\|\Pi_{k-1}\hat G_0\ket\psi\right\|}\right),
\ee
i.e., $\hat G_0\ket\phi$ is the renormalised truncation of $\hat G_0\ket\psi$ at photon number $k-1$.

\qed


\section{Certifying the stellar rank}

\subsection{Achievable fidelities with states of stellar rank 1}
\label{app:core01}

\begin{figure}[h!]
	\begin{center}
		\includegraphics[width=0.45\columnwidth]{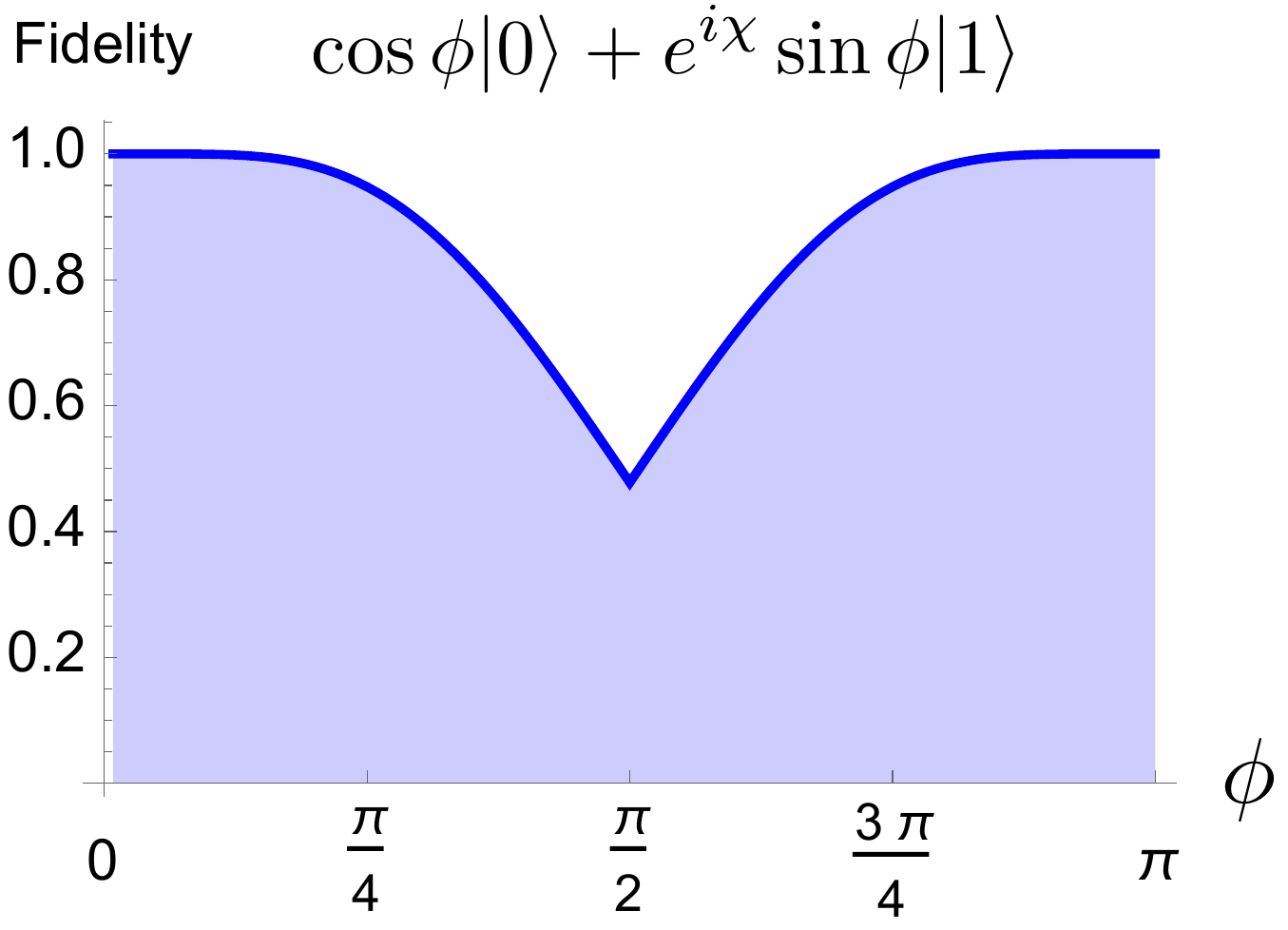}
		\caption{Achievable fidelities with target core states $\cos\phi\ket0+e^{i\chi}\sin\phi\ket1$ using Gaussian states, for all $\phi,\chi\in[0,2\pi]$, as a function of $\phi$.}
		\label{fig:R01}
	\end{center}
\end{figure}

\noindent We have computed numerically the maximum achievable fidelities between Gaussian states and states of the form
\be
\cos\phi\ket0+e^{i\chi}\sin\phi\ket1.
\ee
These are independent of $\chi$ and are depicted in Fig.~\ref{fig:R01} as a function of $\phi$. 

The most robust state of stellar rank $1$ is the Fock state $\ket1$ (and all the states that are obtained from it with Gaussian unitary operations), corresponding to $\phi=\frac\pi2$, as it features the minimal achievable fidelity with Gaussian states.


\subsection{Achievable fidelities with Fock states}
\label{app:RFock}

\noindent In this section, we prove the following result:

\begin{lem} \label{lem:RFock}
Let $k\in\mathbb N^*$ and let $n\in\mathbb N$. Then, the maximum achievable fidelity with the Fock state $\ket n$ using states of stellar rank less than $k$ is given by
\be
\sup_{r^\star(\rho)<k}\braket{n|\rho|n}=\sup_{\xi,\alpha\in\mathbb C}{\sum_{m=0}^{k-1}{|u_{m,n}(\xi,\alpha)|^2}},
\ee
where for all $m\in\{0,\dots,k-1\}$ and all $\xi=re^{i\theta},\alpha\in\mathbb C$, 
\be
u_{m,n}(\xi,\alpha)=\frac1{\sqrt{m!n!c_r}}\left[\partial_z^n(c_rz+s_re^{i\theta}\partial_z-\alpha^*)^me^{-\frac12e^{-i\theta}t_rz^2+\frac\alpha{c_r}z+\frac12e^{i\theta}t_r\alpha^2-\frac12|\alpha|^2}\right]_{z=0},
\ee
with $c_r=\cosh r$, $s_r=\sinh r$ and $t_r=\tanh r$. Moreover, assuming the optimisation yields values $\xi_0,\alpha_0\in\mathbb C$, an optimal approximating state is
\be
\hat D^\dag(\alpha_0)\hat S^\dag(\xi_0)\left(\frac{\Pi_{k-1}\hat S(\xi_0)\hat D(\alpha_0)\ket{n}}{\left\|\Pi_{k-1}\hat S(\xi_0)\hat D(\alpha_0)\ket{n}\right\|}\right).
\ee
\end{lem}

\begin{proof} 

From Theorem~\ref{th:robust} we have
\be
\ba
\sup_{r^\star(\rho)<k}F(\rho,n)&=\sup_{r^\star(\rho)<k}\braket{n|\rho|n}\\
&=\sup_{\hat G\in\mathcal G}{\Tr\left[\Pi_{k-1}\hat G\ket n\bra n\hat G^\dag\right]}\\
&=\sup_{\xi,\alpha\in\mathbb C}{\Tr\left[\Pi_{k-1}\hat D^\dag(\alpha)\hat S^\dag(\xi)\ket{n}\bra{n}\hat S(\xi)\hat D(\alpha)\right]}\\
&=\sup_{\xi,\alpha\in\mathbb C}{\sum_{m=0}^{k-1}\left|\braket{n|\hat S(\xi)\hat D(\alpha)|m}\right|^2},
\ea
\label{robusku}
\ee
where we used the fact that any single-mode Gaussian unitary operation may be decomposed as a squeezing and a displacement. Assuming the optimisation yields values $\xi_0,\alpha_0\in\mathbb C$, the optimal core state used in the approximation is
\be
\ket C=\frac{\Pi_{k-1}\hat S(\xi_0)\hat D(\alpha_0)\ket{n}}{\left\|\Pi_{k-1}\hat S(\xi_0)\hat D(\alpha_0)\ket{n}\right\|}.
\label{coreopti}
\ee
Now for all $m\in\{0,\dots,k-1\}$,
\be
\ba
\braket{n|\hat S(\xi)\hat D(\alpha)|m}&=\frac1{\sqrt{m!n!}}\braket{0|\hat a^n\hat S(\xi)\hat D(\alpha)(\hat a^\dag)^m|0}\\
&=\frac1{\sqrt{m!n!}}\braket{0|\hat a^n(c_r\hat a^\dag+s_re^{i\theta}\hat a-\alpha^*)^m\hat S(\xi)\hat D(\alpha)|0},
\ea
\label{FockSD}
\ee
with $c_r=\cosh r$, $s_r=\sinh r$. Hereafter we also set $t_r=\tanh r$. We have $\braket{0|\chi}=F^\star_\chi(0)$ for all states $\chi$, hence switching to the stellar representation~\cite{chabaud2020stellar} we obtain
\be
\braket{n|\hat S(\xi)\hat D(\alpha)|m}=\frac1{\sqrt{m!n!c_r}}\left[\partial_z^n\left(c_rz+s_re^{i\theta}\partial_z-\alpha^*\right)^me^{-\frac12e^{-i\theta}t_rz^2+\frac\alpha{c_r}z+\frac12e^{i\theta}t_r\alpha^2-\frac12|\alpha|^2}\right]_{z=0}.
\ee
Setting, for all $m\in\{0,\dots,k-1\}$, $u_{m,n}(\xi,\alpha)=\braket{n|\hat S(\xi)\hat D(\alpha)|m}$, we finally obtain with Eq.~(\ref{robusku}),
\be
\sup_{r^\star(\rho)<k}\braket{n|\rho|n}=\sup_{\xi,\alpha\in\mathbb C}{\sum_{m=0}^{k-1}{|u_{m,n}(\xi,\alpha)|^2}},
\ee
where for all $m\in\{0,\dots,k-1\}$ and all $\xi=re^{i\theta},\alpha\in\mathbb C$, 
\be
u_m(\xi,\alpha)=\frac1{\sqrt{m!n!c_r}}\left[\partial_z^n(c_rz+s_re^{i\theta}\partial_z-\alpha^*)^me^{-\frac12e^{-i\theta}t_rz^2+\frac\alpha{c_r}z+\frac12e^{i\theta}t_r\alpha^2-\frac12|\alpha|^2}\right]_{z=0},
\ee
with $c_r=\cosh r$, $s_r=\sinh r$ and $t_r=\tanh r$. Moreover, assuming the optimisation yields values $\xi_0,\alpha_0\in\mathbb C$, an optimal approximating state is $\ket\phi=\hat D^\dag(\alpha_0)\hat S^\dag(\xi_0)\ket C$, where $\ket C$ is defined in Eq.~(\ref{coreopti}). Namely,
\be
\ket\phi=\hat D^\dag(\alpha_0)\hat S^\dag(\xi_0)\left(\frac{\Pi_{k-1}\hat S(\xi_0)\hat D(\alpha_0)\ket{n}}{\left\|\Pi_{k-1}\hat S(\xi_0)\hat D(\alpha_0)\ket{n}\right\|}\right),
\ee
i.e., $\hat S(\xi_0)\hat D(\alpha_0)\ket\phi$ is the renormalised truncation of $\hat S(\xi_0)\hat D(\alpha_0)\ket{C_\psi}$ at photon number $k-1$.

\end{proof}

\noindent From this expression we have obtained numerically the maximum achievable fidelities up to $n=5$. The table below summarises the results, with the convention $r=k+1$.

\begin{table}
\centering
\setlength\tabcolsep{0.1pt}
\bgroup
\def\arraystretch{2}
\begin{tabular}{| c | c | c | c | c | c | c |}
\hline
\diagbox[height=7ex,width=5em]{$\quad\; n$}{$r\;\;$} & \,\quad0\quad\, & \,\quad1\quad\, & \,\quad2\quad\, & \,\quad3\quad\, & \,\quad4\quad\, & \,\quad5\quad\,\\
\hline
0 & 1 & 1 & 1 & 1 & 1 & 1\\
\hline
1 & 0.478 & 1 & 1 & 1 & 1 & 1\\
\hline
2 & 0.381 & 0.557 & 1 & 1 & 1 & 1\\
\hline
3 & 0.333 & 0.462 & 0.593 & 1 & 1 & 1\\
\hline
4 & 0.301 & 0.409 & 0.501 & 0.612 & 1 & 1\\
\hline
5 & 0.279 & 0.374 & 0.449 & 0.525 & 0.626 & 1\\
\hline
\end{tabular}
\egroup
\caption{Maximum achievable fidelities with Fock states $\ket n$ using states with stellar rank less or equal to $r$, for $n,r\in\{0,\dots,5\}$.}
\label{tab:probaest}
\end{table}

\end{document}